\documentclass[11pt]{article}
\usepackage[a4paper,margin=1in]{geometry}
\usepackage{amsmath,amssymb,amsthm}
\usepackage{booktabs}
\usepackage{enumitem}
\usepackage[pdfusetitle,hidelinks]{hyperref}
\usepackage{xspace}
\usepackage{tabularx}

\newcommand{\VADAR}{VA-DAR\xspace}
\newcommand{\DID}{\mathsf{DiscoveryID}}
\newcommand{\CID}{\mathsf{CID}}
\newcommand{\REV}{\mathsf{REV}}
\newcommand{\Commit}{\mathsf{commit}}
\newcommand{\SA}{\mathsf{SA}}
\newcommand{\Norm}{\mathsf{Norm}}
\newcommand{\Zeroize}{\mathsf{Zeroize}}
\newcommand{\concat}{\,\|\,}

\newtheorem{definition}{Definition}
\newtheorem{theorem}{Theorem}

\newtheorem{remark}{Remark}
\newtheorem{corollary}{Corollary}

\title{\textbf{VA-DAR: A Vendor-Agnostic Deterministic Artifact Resolution\\
for Serverless, Enumeration-Resistant Wallet Recovery}}
\author{
  \texttt{Jian Sheng Wang} \\
  \texttt{jason@acechain.io} \\
  Montreal, Canada
}
\date{November 24, 2026}

\begin{document}
\maketitle

\begin{abstract}
Serverless wallet recovery can use a public registry for decentralized lookup,
but naive identifier hashing in such a registry leaks membership through
enumeration.
We present \VADAR, a keyed-discovery protocol for ACE-GF--based wallets
that use device-bound passkeys for day-to-day local unlock while supporting
cross-device recovery using only a user-provided identifier (e.g., email) and
a single recovery passphrase (the \emph{Primary Passphrase}).
As a discovery-and-recovery layer over ACE-GF, \VADAR inherits ACE-GF's
context-isolated, algorithm-agile key derivation, which allows migration to
post-quantum algorithms at the identity layer without disruption.

The protocol maps a privacy-preserving discovery identifier to an immutable content identifier
of a backup sealed artifact stored on a decentralized storage network.
Concretely, a user derives passphrase-rooted key material with a memory-hard
KDF, domain-separates keys for artifact sealing and discovery indexing, and
publishes a registry record keyed by a passphrase-derived discovery identifier.

\VADAR provides:
(i)~cross-device recovery using only identifier and passphrase,
(ii)~computational resistance to public-directory enumeration,
(iii)~integrity of discovery mappings via owner authorization, and
(iv)~rollback/tamper detection via monotonic versioning and artifact
commitments.
We define three sealed artifact roles, two update-authorization options,
and three protocol flows (registration, recovery, update).
We formalize security goals via cryptographic games and derive theorem-level
security bounds, under
standard assumptions (memory-hard KDF hardness, keyed-function
pseudorandomness, signature unforgeability, and hash collision resistance),
that \VADAR meets these goals while remaining vendor-agnostic and
chain-agnostic.
End-to-end post-quantum deployment is discussed as a migration path and
requires a PQ-secure instantiation of registry authorization.
\end{abstract}

\section{Introduction}
\label{sec:intro}

Consumer wallets traditionally rely on mnemonic seed phrases for recovery,
creating usability and custody failures.
Meanwhile, passkeys (hardware-backed credentials) provide strong local
authentication but are commonly coupled with vendor synchronization or
device-bound state, which can conflict with a strict ``no centralized
server'' requirement.

This paper takes a \emph{hybrid} approach:
\begin{itemize}[nosep,leftmargin=1.5em]
  \item \textbf{Local security and convenience.}
    A device-bound Passkey PRF (hardware-gated) seals a local artifact for
    biometric unlock on a given device.
  \item \textbf{Cross-device recovery without centralized servers.}
    A backup artifact is sealed \emph{only} under a user-chosen recovery
    passphrase and stored in decentralized storage.
    A public registry maps a \emph{keyed} discovery identifier to the
    artifact's content identifier.
\end{itemize}

\VADAR avoids public enumeration by requiring knowledge of the recovery passphrase to
compute the discovery identifier used for registry lookup.

\paragraph{Non-goals.}
\VADAR does not, by itself, revoke already-derived on-chain keys or
invalidate prior addresses.
Key revocation requires account migration or smart-account rotation
mechanisms and is outside the scope of this document.

\subsection*{Contributions}
\begin{itemize}[nosep,leftmargin=1.5em]
  \item A keyed discovery construction for public registries that resists
        identifier-only enumeration.
  \item A domain-separated key schedule that prevents cross-role key reuse
        between discovery, sealing, and registry authorization.
  \item A deployable authorization model (owner-signature updates) and an
        extension path for passphrase-derived authorization.
  \item A versioned, commitment-bound registry update protocol with
        owner authorization and rollback resistance.
  \item Formal security games and theorem-level guarantees for enumeration
        resistance, mapping integrity, and rollback safety.
\end{itemize}

\subsection*{Scope and Novelty Boundary}
This paper's novelty is at the protocol-composition layer:
how passphrase-gated keyed discovery, decentralized artifact resolution,
owner-authorized registry mutation, and version/commit consistency are
combined and formalized in one end-to-end model for serverless recovery.
We do not claim novelty for Argon2id, HKDF, HMAC, AEAD, or signatures as
standalone primitives.

\subsection*{Relation to Prior ACE-GF Work}
ACE-GF~\cite{wang-acegf-2025} is treated as a publicly available primitive
specification providing
sealed artifact handling and deterministic derivation interfaces.
This paper contributes the decentralized discovery-and-recovery layer,
keyed discovery identifier design, and registry integrity/rollback model.
It does not alter ACE-GF internals.

\subsection*{Positioning vs.\ Adjacent Approaches}
The most widely deployed recovery baseline remains mnemonic backup
(e.g., BIP-39)~\cite{bip39}, which is portable but operationally fragile for
mainstream users.
Passkeys/WebAuthn improve local authentication UX and phishing resistance, but
their core standard scope is authenticator-backed login, not by itself a
complete decentralized wallet recovery protocol~\cite{webauthn}.
Account-abstraction ecosystems (e.g., ERC-4337 style smart accounts) provide
account-layer controls such as key rotation/recovery policies, but
they do not directly solve public-directory enumeration risk for identifier
lookup in serverless backup discovery~\cite{eip4337}.
VA-DAR is positioned as a complementary layer focused on passphrase-gated
discovery and commitment/version-bound backup resolution.

\section{System Model and Preliminaries}
\label{sec:model}

\subsection{Entities}
\begin{itemize}[nosep,leftmargin=1.5em]
  \item \textbf{Client Device}~$D_i$: runs the wallet application.
  \item \textbf{Passkey PRF}~$\mathrm{PRF}_{D_i}(\cdot)$: a device-bound
        pseudo-random function gated by local user presence/biometrics
        (e.g., via a platform PRF extension).
  \item \textbf{Decentralized Storage}: content-addressed storage such as
        Arweave~\cite{arweave} providing immutable content identifiers.
  \item \textbf{Public Registry}: a public, append-only or updatable
        mapping layer (e.g., L2 smart contract) storing discovery records.
  \item \textbf{Relayer} (optional): a convenience component that submits
        registry transactions on behalf of clients.
\end{itemize}

\subsection{Secrets and Artifacts}
\begin{itemize}[nosep,leftmargin=1.5em]
  \item $\REV$ (Root Entity Value): the root secret from which the wallet
        derives cryptographic identities and keys.
  \item \textbf{Primary Passphrase}~$P$: user-known recovery passphrase
        (portable secret).
  \item $\SA$ (Sealed Artifact): a binary artifact encrypting $\REV$ under
        a derived sealing key.
\end{itemize}

\subsection{Cryptographic Primitives}
We assume:
\begin{itemize}[nosep,leftmargin=1.5em]
  \item Memory-hard password KDF: $\mathsf{Argon2id}(\cdot)$~\cite{argon2}.
  \item KDF for domain separation:
        $\mathsf{HKDF}(\cdot)$~\cite{hkdf}.
  \item Message authentication:
        $\mathsf{HMAC}(\cdot)$~\cite{hmac}.
  \item Authenticated encryption:
        $\mathsf{AEAD.Enc}/\mathsf{AEAD.Dec}$ with associated
        data~\cite{aead-siv}.
  \item Digital signatures:
        $\mathsf{Sign}/\mathsf{Verify}$ (EUF-CMA secure)~\cite{euf-cma}.
  \item Collision-resistant hash function: $H(\cdot)$.
\end{itemize}

\subsection{Notation}
\begin{itemize}[nosep,leftmargin=1.5em]
  \item $\concat$ denotes concatenation.
  \item $\Norm(\mathtt{email})$ denotes canonical normalization of an
        identifier (lowercase, trimming whitespace, etc.).
  \item $\Zeroize(\cdot)$ denotes secure memory clearing.
\end{itemize}

\section{Threat Model and Security Goals}
\label{sec:threat}

\subsection{Adversary Capabilities}
The adversary~$\mathcal{A}$ is probabilistic polynomial-time (PPT) and can:
\begin{itemize}[nosep,leftmargin=1.5em]
  \item Observe the full public registry state and all on-chain data.
  \item Query the registry arbitrarily often with arbitrary identifiers.
  \item Attempt offline guessing attacks if it obtains a sealed artifact
        ($\SA2$) from decentralized storage.
  \item Attempt to front-run or overwrite registry records, subject to
        registry authorization rules.
  \item Submit competing transactions, replay stale updates, and tamper
        with storage payload delivery.
  \item Control relayers and network scheduling.
  \item Perform denial-of-service by spamming queries or refusing to
        store/pin content.
\end{itemize}
We do \emph{not} assume a trusted centralized server for rate-limiting,
identity verification, or content hosting.
The adversary does not know honest users' passphrases unless by
online/offline guessing, and cannot break underlying cryptographic
assumptions.
We also do not assume global uniqueness of a registry deployment:
an adversary can deploy a parallel registry with copied public metadata
(including $\mathtt{app\_id}$).
Security therefore requires clients to authenticate a canonical
$\mathtt{registry\_id}$ rather than relying on namespace secrecy.

\subsection{Security Goals}
\begin{description}[nosep,leftmargin=0em]
  \item[G1: Cross-device recoverability.]
    A legitimate user can recover $\REV$ on a new device using only
    $(I, P)$ and publicly available infrastructure.
  \item[G2: Enumeration resistance.]
    Observers should not be able to learn wallet membership solely by
    knowing an identifier such as an email; the registry should not be a
    public directory.
  \item[G3: Integrity of discovery mapping.]
    Unauthorized parties cannot update a victim's mapping to point to
    attacker-chosen content.
  \item[G4: Rollback resistance.]
    The system should resist reverting the mapping to older CIDs without
    explicit authorization.
  \item[G5: Minimization of device lock-in.]
    Device-bound PRF improves local security but must not be required for
    cross-device recovery.
\end{description}

\section{Design Overview: Three Sealed Artifacts}
\label{sec:overview}

\VADAR formalizes three artifact roles in the wallet lifecycle:

\begin{table}[h]
\centering
\begin{tabular}{@{}llp{6cm}@{}}
\toprule
Artifact & Sealing Factor & Purpose \\
\midrule
$\SA1$ (Local) & $\mathrm{PRF}_{D_1}$-derived &
  Daily unlock on device $D_1$ (biometrics). \\
$\SA2$ (Backup) & Primary Passphrase $P$ only &
  Cross-device recovery (portable). \\
$\SA3$ (New Local) & $\mathrm{PRF}_{D_2}$-derived &
  Daily unlock on new device $D_2$. \\
\bottomrule
\end{tabular}
\caption{Sealed artifact roles.}
\label{tab:artifacts}
\end{table}

$\SA2$ must be decryptable without any device-bound PRF; otherwise recovery
would fail on device changes under a ``no vendor sync'' constraint (goal~G5).

\section{Key Derivation with Strict Domain Separation}
\label{sec:keys}

\subsection{Two-Stage Passphrase Derivation}
To make lookup possible \emph{before} fetching $\SA2$, \VADAR uses two
separate derivations from the same passphrase~$P$.

\paragraph{Stage A (lookup/auth root; computable from $(I,P)$ only).}
Define a deterministic lookup salt:
\begin{equation}\label{eq:salt-lookup}
  \mathsf{salt}_{\mathrm{lookup}} \leftarrow
  H\!\bigl(\texttt{"va-dar:lookup:v1"} \concat \Norm(I)\bigr).
\end{equation}
Then derive:
\begin{equation}\label{eq:klookup}
  K_{\mathrm{lookup}} \leftarrow
  \mathsf{Argon2id}(P, \mathsf{salt}_{\mathrm{lookup}},
  \mathsf{params}_{\mathrm{lookup}}).
\end{equation}

\paragraph{Stage B (artifact decryption root; bound to artifact-local salt).}
After fetching $\SA2$, parse $\mathsf{salt}_{pw}, \mathsf{params}_{pw}$ from
the artifact and derive:
\begin{equation}\label{eq:ksealroot}
  K_{\mathrm{sealroot}} \leftarrow
  \mathsf{Argon2id}(P, \mathsf{salt}_{pw}, \mathsf{params}_{pw}).
\end{equation}

\subsection{Separate Keys for Sealing, Indexing, and Authorization}
To ensure safe reuse of a single passphrase~$P$ across protocol roles
(artifact sealing, discovery indexing, registry authorization), the client
derives independent keys:
\begin{align}
  K_{\mathrm{sa}}  &\leftarrow
    \mathsf{HKDF}(K_{\mathrm{sealroot}},
    \mathsf{info}=\texttt{"acegf:sa2:seal"}),
    \label{eq:ksa} \\
  K_{\mathrm{idx}} &\leftarrow
    \mathsf{HKDF}(K_{\mathrm{lookup}},
    \mathsf{info}=\texttt{"va\text{-}dar:discovery:index"}),
    \label{eq:kidx} \\
  K_{\mathrm{reg}} &\leftarrow
    \mathsf{HKDF}(K_{\mathrm{lookup}},
    \mathsf{info}=\texttt{"va\text{-}dar:registry:auth"}).
    \label{eq:kreg}
\end{align}
This domain separation is normative.
Implementations \textsc{must not} reuse a single key directly for both
AEAD sealing and registry indexing or authorization.

\section{Discovery Identifier Construction}
\label{sec:did}

\subsection{Keyed Discovery Identifier}
Let $I$ be a user-provided identifier such as an email address.
Let $\mathsf{ctx}_{\mathrm{did}}$ be public domain-separation context
(e.g., application id, registry id, protocol version).
Define:
\begin{equation}\label{eq:did}
  \DID \leftarrow
  \mathsf{HMAC}(K_{\mathrm{idx}},
  \mathsf{ctx}_{\mathrm{did}} \concat \Norm(I)).
\end{equation}

\paragraph{Rationale.}
Unlike $H(I \concat \mathsf{global\_salt})$, the keyed construction
prevents third parties from computing $\DID$ without knowledge of~$P$,
which gives enumeration resistance~(G2).

\subsection{Length and Encoding}
$\DID$ is a binary string of the HMAC output length.
For registry usage, it may be encoded as hex or base64url.
The normalization function~$\Norm$ and context
$\mathsf{ctx}_{\mathrm{did}}$ must be deterministic and specified;
inconsistent settings cause account fragmentation.
Implementations may fix a structure for $\mathsf{ctx}_{\mathrm{did}}$,
e.g., $\mathtt{app\_id} \concat \mathtt{registry\_id} \concat
\mathtt{version}$.
Normatively, $\mathtt{registry\_id}$ should include at least
$(\mathtt{chain\_id}, \mathtt{contract\_address})$ for on-chain registries.
The $\mathtt{app\_id}$ is public and \emph{not} an authenticity secret;
authenticity is anchored by client-verified canonical
$\mathtt{registry\_id}$.

\section{Sealed Artifact Formats}
\label{sec:formats}

\subsection{Backup Artifact \texorpdfstring{$\SA2$}{SA2} (Password-Sealed)}
$\SA2$ must include all parameters needed for decryption and verification,
except the passphrase.
A minimal format is:

\begin{table}[h]
\centering
\begin{tabular}{@{}lp{8cm}@{}}
\toprule
Field & Description \\
\midrule
\texttt{version}    & Artifact format version \\
\texttt{salt\_pw}   & Salt for $\mathsf{Argon2id}$ \\
\texttt{params\_pw} & $\mathsf{Argon2id}$ parameters $(m,t,p)$ \\
\texttt{aead\_alg}  & AEAD algorithm identifier \\
\texttt{nonce}      & AEAD nonce \\
\texttt{aad}        & Associated data (context, metadata) \\
\texttt{ct}         & $\mathsf{AEAD.Enc}(K_{\mathrm{sa}}, \mathsf{nonce},
                       \REV, \mathsf{aad})$ (nonce from row above) \\
\bottomrule
\end{tabular}
\caption{$\SA2$ record format.}
\label{tab:sa2}
\end{table}

\subsection{Local Artifacts \texorpdfstring{$\SA1$/$\SA3$}{SA1/SA3} (PRF-Sealed)}
Local artifacts are sealed under a device-bound PRF-derived secret:
\begin{equation}\label{eq:klocal}
  K_{\mathrm{local}} \leftarrow
    \mathsf{HKDF}\bigl(\mathrm{PRF}_{D_i}(\mathsf{label}),\,
    \mathsf{info}=\texttt{"acegf:local:seal"}\bigr).
\end{equation}
The exact label is implementation-defined but should be stable per-app and
per-user on the device.
Local artifacts \textsc{must not} be required for cross-device
recovery~(G5).

\section{Registry Record and Update Authorization}
\label{sec:registry}

\subsection{Registry Record Structure}
The public registry stores a mapping from $\DID$ to a record:

\begin{table}[h]
\centering
\begin{tabular}{@{}lp{8cm}@{}}
\toprule
Field & Description \\
\midrule
\texttt{cid}    & Content identifier of $\SA2$ in decentralized storage \\
\texttt{ver}    & Monotonic version counter \\
\texttt{commit} & Commitment to artifact, e.g., $H(\SA2)$ \\
\texttt{pk\_owner} & Owner verification key bound at first registration \\
\texttt{auth}   & Owner signature on $(\mathtt{op},\DID,\mathtt{cid},\mathtt{ver},\Commit)$ \\
\bottomrule
\end{tabular}
\caption{Registry record structure.}
\label{tab:record}
\end{table}

\subsection{Update Authorization (Primary Instantiation)}
\label{sec:authopt}
Because no centralized server exists, update authorization must be
cryptographically verifiable.

\paragraph{Signature-based owner authorization (Option~B).}
Store an owner public key~$\mathsf{pk}_{\mathrm{owner}}$ in the record on
first registration, and enforce that first-write key binding is atomic
(initialize-if-empty semantics on $(\DID, \mathsf{pk}_{\mathrm{owner}})$).
Every state-changing operation must be signed:
\begin{equation}\label{eq:auth-sig}
  \sigma \leftarrow
    \mathsf{Sign}\bigl(\mathsf{sk}_{\mathrm{owner}},\,
    \mathtt{op} \concat \DID \concat \mathtt{cid} \concat \mathtt{ver} \concat
    \Commit\bigr),
\end{equation}
where $\mathtt{op} \in \{\texttt{init}, \texttt{update},
\texttt{tombstone}\}$ is a fixed-length operation tag, all fields use a
fixed-length or length-prefixed encoding, and a tombstone carries the
redirect target $\DID'$ (or $\bot$) in the $\mathtt{cid}$ field and consumes
the next version number.
A signature for one operation therefore cannot be accepted as another.

\paragraph{Front-running-resistant first registration.}
Initialize-if-empty alone does not stop front-running: $\DID$ becomes
public when the initializing transaction is broadcast, and an adversary
could submit its own $\texttt{init}$ for that $\DID$ with an attacker key
and have it ordered first.
\VADAR therefore initializes a record by commit--reveal.
The client first submits
$c \leftarrow H(\texttt{"va-dar:init:v1"} \concat \DID \concat
\mathsf{pk}_{\mathrm{owner}} \concat r)$ for a fresh random $r$ of
$\lambda$ bits; the registry stores $c$ with its inclusion height.
After at least $\Delta$ blocks the client reveals
$(\DID, \mathsf{pk}_{\mathrm{owner}}, r)$ together with the signed
$\texttt{init}$ record.
The registry accepts the reveal only if $\DID$ is empty, a stored
commitment equals $H(\texttt{"va-dar:init:v1"} \concat \DID \concat
\mathsf{pk}_{\mathrm{owner}} \concat r)$ and is at least $\Delta$ blocks
old, and $\sigma$ verifies under the revealed
$\mathsf{pk}_{\mathrm{owner}}$.
The commitment hides $\DID$, so an adversary that learns $\DID$ only from
the reveal can hold no commitment older than the reveal itself, and its own
reveal cannot be accepted before the honest one, provided the honest reveal
is included within $\Delta$ blocks of broadcast.
An adversary that computes $\DID$ in advance must guess $P$ at Stage~A
cost.
A client whose reveal is rejected learns at registration time that $\DID$
is taken and must not rely on that entry for recovery.
On-chain verification is a single signature check; the cost is owner-key
lifecycle management.
Two deployment models are common:
(i) a randomly generated owner key protected by secure backup, or
(ii) deterministic derivation from passphrase-derived material
(higher UX, but security then additionally depends on passphrase-guessing
cost/entropy).

\begin{remark}
This paper analyzes Option~B as the primary deployment target.
A passphrase-derived HMAC authorization variant (Option~A) can be considered
only in verifier models that can validate HMAC-based statements without
exposing secret material (e.g., strong off-chain verifier assumptions or
sound ZK validity systems); we do not rely on Option~A in the core theorems.
\end{remark}

\subsection{Rollback Resistance}
\label{sec:rollback-mech}
The registry must enforce $\mathtt{ver}$ as monotonically increasing per
$\DID$ (e.g., only accept updates where
$\mathtt{ver}_{\mathrm{new}} > \mathtt{ver}_{\mathrm{old}}$).
This resists rollback to an older CID absent key compromise.

\section{Protocol Flows}
\label{sec:flows}

\subsection{Flow~1: Initial Registration (Device \texorpdfstring{$D_1$}{D1})}
\label{sec:flow-reg}

\begin{enumerate}[nosep,leftmargin=2em]
  \item Device $D_1$ holds $\REV$ and local $\SA1$ sealed under
        $\mathrm{PRF}_{D_1}$.
  \item User chooses Primary Passphrase~$P$.
  \item $D_1$ unlocks $\SA1$ via passkey biometrics, reconstructs $\REV$
        in memory.
  \item $D_1$ obtains $\mathsf{salt}_{\mathrm{lookup}}$ from $I$ via
        Eq.~\eqref{eq:salt-lookup}, then derives $K_{\mathrm{lookup}}$ from
        $(I,P)$ (Eq.~\eqref{eq:klookup}) and derives
        $K_{\mathrm{idx}}$.
  \item $D_1$ samples fresh $\mathsf{salt}_{pw}$ and computes
        $K_{\mathrm{sealroot}}$ (Eq.~\eqref{eq:ksealroot}), then
        derives $K_{\mathrm{sa}}$.
  \item $D_1$ constructs $\SA2$ by sealing $\REV$ under $K_{\mathrm{sa}}$.
  \item $D_1$ uploads $\SA2$ to decentralized storage, obtaining $\CID$.
  \item $D_1$ computes
        $\DID = \mathsf{HMAC}(K_{\mathrm{idx}},
        \mathsf{ctx}_{\mathrm{did}} \concat \Norm(I))$.
  \item $D_1$ prepares owner authorization material:
        either a random owner keypair or a deterministic owner keypair
        derived from passphrase-rooted material (deployment choice).
  \item $D_1$ writes registry record:
        $\DID \mapsto \{\mathtt{cid}, \mathtt{ver}=1,
        \Commit, \mathsf{pk}_{\mathrm{owner}}, \sigma\}$,
        atomically binding $\mathsf{pk}_{\mathrm{owner}}$ on first write
        through the commit--reveal initialization of
        Section~\ref{sec:authopt} (commit, wait $\Delta$ blocks, reveal
        with the signed $\texttt{init}$ record).
  \item $D_1$ calls $\Zeroize$ on $\REV$ and all derived keys in memory.
\end{enumerate}

\subsection{Flow~2: Recovery (Device \texorpdfstring{$D_2$}{D2})}
\label{sec:flow-recover}

\begin{enumerate}[nosep,leftmargin=2em]
  \item User installs wallet on $D_2$, inputs identifier $I$ and
        Primary Passphrase~$P$.
  \item $D_2$ obtains $\mathsf{salt}_{\mathrm{lookup}}$ from $I$ via
        Eq.~\eqref{eq:salt-lookup}, then derives $K_{\mathrm{lookup}}$ from
        $(I,P)$ (Eq.~\eqref{eq:klookup}) and derives $K_{\mathrm{idx}}$.
  \item $D_2$ computes
        $\DID = \mathsf{HMAC}(K_{\mathrm{idx}},
        \mathsf{ctx}_{\mathrm{did}} \concat \Norm(I))$.
  \item $D_2$ queries registry for $\DID$, obtains $\CID$ and record
        metadata.
  \item $D_2$ fetches $\SA2$ from decentralized storage using $\CID$.
  \item $D_2$ verifies $H(\SA2) = \Commit$.
  \item $D_2$ parses $(\mathsf{salt}_{pw}, \mathsf{params}_{pw})$ from
        $\SA2$, derives $K_{\mathrm{sealroot}}$ and then $K_{\mathrm{sa}}$.
  \item $D_2$ decrypts $\SA2$ using $K_{\mathrm{sa}}$, reconstructs
        $\REV$ in memory.
  \item $D_2$ enables passkey and seals a new local $\SA3$ under
        $\mathrm{PRF}_{D_2}$.
  \item $D_2$ derives chain-specific keys and addresses from $\REV$ as
        needed.
  \item $D_2$ calls $\Zeroize$ on $\REV$ after key derivations complete.
\end{enumerate}

\subsection{Flow~3: Backup Update / Discovery Revocation}
\label{sec:flow-update}

This flow updates the registry mapping to a new CID (e.g., password rotation,
parameter strengthening, or periodic hygiene).
It provides \emph{discovery revocation} (directory update), not
cryptographic invalidation of previously derived on-chain keys.

\begin{enumerate}[nosep,leftmargin=2em]
  \item User on an authorized device unlocks local artifact ($\SA1$ or
        $\SA3$) via passkey and reconstructs $\REV$ in memory.
  \item User selects target passphrase $P'$ (either unchanged $P'=P$ or
        rotated $P'\neq P$).
  \item Construct new backup artifact $\SA2'$ (fresh
        $\mathsf{salt}_{pw}'$) sealed under~$P'$, upload, obtaining $\CID'$.
  \item \textbf{Case 1: passphrase unchanged} ($P' = P$).
        $\DID$ remains unchanged; submit authorized update on the same key
        with $\mathtt{ver}\leftarrow\mathtt{ver}+1$.
  \item \textbf{Case 2: passphrase rotated} ($P' \neq P$).
        Compute $\DID'$ from $(I,P')$ and create a \emph{new} registry
        entry at key $\DID'$ (its own version sequence, starting at~1),
        initialized by commit--reveal as in Flow~1.
  \item In Case 2, mark prior $\DID$ entry as tombstoned (or set an
        optional redirect pointer to $\DID'$) with explicit migration
        metadata and policy-defined grace period, by a
        $\texttt{tombstone}$ operation signed as in
        Eq.~\eqref{eq:auth-sig}.
  \item $\Zeroize$ $\REV$ from memory.
\end{enumerate}

\begin{remark}
Version monotonicity is enforced \emph{per registry key}.
When $P$ changes, $\DID$ changes, so continuity must be modeled as a key
migration workflow (old-key tombstone/redirect), not as an in-place version
increment on the old key.
\end{remark}

\section{Formal Security Definitions}
\label{sec:security}

We formalize the security goals from Section~\ref{sec:threat} as
cryptographic games.
In all games the adversary~$\mathcal{A}$ is PPT, and authorization follows
the signature-based registry model in Section~\ref{sec:authopt}.

\begin{definition}[Enumeration-Resistance Game $\mathbf{G}_{\mathrm{enum}}$]
\label{def:enum}
The challenger samples a bit $b \gets \{0,1\}$.
If $b=1$, the challenger registers a target identifier $I^*$ with a
passphrase sampled from a distribution with minimum entropy
$H_\infty(P)\ge \mu$, producing
$\DID^* = \mathsf{HMAC}(K_{\mathrm{idx}}^*,\,
\mathsf{ctx}_{\mathrm{did}} \concat \Norm(I^*))$.
If $b=0$, $\DID^*$ is drawn uniformly at random from the output space,
and the challenger publishes a record under it whose other fields
are generated as in the $b=1$ case.
$\mathcal{A}$ is given public registry access (including the record under
$\DID^*$), the identifier $I^*$, and oracle access to
$\mathsf{Argon2id}$ with at most $q_A$ queries, but \emph{not} the
passphrase.
$\mathcal{A}$ outputs a guess $b'$, and
$\mathrm{Adv}^{\mathrm{enum}}_{\mathcal{A}} =
\bigl|2\Pr[b'=b]-1\bigr|$.
\end{definition}

\begin{definition}[Mapping-Integrity Game $\mathbf{G}_{\mathrm{map}}$]
\label{def:map}
The challenger registers a record for $\DID^*$ with authorization
material~$\mathsf{auth}^*$ (signature under
$\mathsf{pk}_{\mathrm{owner}}^*$ in the primary model).
The game includes first-registration security: if $\DID^*$ is initially
empty, only an authorized initialize operation that atomically binds
$\mathsf{pk}_{\mathrm{owner}}^*$ is accepted.
The registry is assumed to provide atomic conditional writes
(initialize-if-empty) with standard ledger ordering guarantees, and
initialization uses the commit--reveal rule of Section~\ref{sec:authopt}.
$\mathcal{A}$ observes every broadcast transaction, including the
challenger's commitment and reveal, and may order its own transactions
before them, but each honest transaction is included within $\Delta$ blocks
of broadcast.
$\mathcal{A}$ has oracle access to $\mathsf{Argon2id}$ with at most $q_A$
queries.
$\mathcal{A}$ is given full registry access and may issue arbitrary
create/update/delete requests for any $\DID \neq \DID^*$.
$\mathcal{A}$ wins if it causes the registry to accept a state-changing
operation on $\DID^*$ (create overwrite, update, or delete) without
producing valid authorization under the victim's key material.
\end{definition}

\begin{definition}[Rollback-Safety Game $\mathbf{G}_{\mathrm{roll}}$]
\label{def:roll}
The challenger registers and updates a record for $\DID^*$ through
versions $1, 2, \ldots, n$ with corresponding CIDs
$c_1, c_2, \ldots, c_n$ and commitments $h_1, \ldots, h_n$.
$\mathcal{A}$ controls the network and storage layer.
$\mathcal{A}$ may delay or censor messages, but the client is assumed to
query a finalized/fresh registry view (standard ledger finality assumption).
$\mathcal{A}$ wins if a client performing recovery on version~$n$
accepts an artifact corresponding to version $j < n$ (stale CID replay)
or accepts an artifact $\SA2'$ such that $H(\SA2') \neq h_n$
(tampering), without the client detecting the discrepancy.
\end{definition}

\section{Main Results}
\label{sec:results}

\begin{theorem}[Enumeration Resistance]
\label{thm:enum}
Model $\mathsf{Argon2id}$ at cost $(\mathsf{params}_{\mathrm{lookup}})$
as a random oracle, let the passphrase have min-entropy $\mu$, and
assume $\mathsf{HMAC}(K_{\mathrm{idx}}, \cdot)$ is a PRF when keyed by
$K_{\mathrm{idx}}$ derived from $K_{\mathrm{lookup}}$ via $\mathsf{HKDF}$.
Then for any adversary~$\mathcal{A}$ playing
$\mathbf{G}_{\mathrm{enum}}$ with $q_A$ Stage-A evaluations,
\[
  \mathrm{Adv}^{\mathrm{enum}}_{\mathcal{A}}
  \;\leq\;
  \frac{q_A}{2^{\mu}}
  \;+\;
  \mathrm{Adv}^{\mathrm{prf}}_{\mathsf{HMAC}}
  \;+\; \mathsf{negl}(\lambda).
\]
That is, identifier-only enumeration reduces to passphrase guessing with
per-guess memory-hard cost and passphrase entropy parameter~$\mu$.
\end{theorem}

\begin{proof}[Proof sketch with game hops]
Consider games $G_0,G_1,G_2$.
$G_0$ is the real experiment $\mathbf{G}_{\mathrm{enum}}$.
In $G_1$, replace
$K_{\mathrm{lookup}}=\mathsf{Argon2id}(P,\mathsf{salt}_{\mathrm{lookup}},\cdot)$
with a uniform random key.
$G_0$ and $G_1$ differ only if $\mathcal{A}$ queries the oracle on
$(P,\mathsf{salt}_{\mathrm{lookup}})$, which happens with probability at
most $q_A/2^{\mu}$.
In $G_2$, replace
$\mathsf{HMAC}(K_{\mathrm{idx}},
\mathsf{ctx}_{\mathrm{did}}\concat\Norm(I^*))$
with a uniform random string.
Any distinguisher between $G_1$ and $G_2$ yields a PRF distinguisher for
HMAC under HKDF-derived keying.
In $G_2$, the challenge transcript is independent of $b$, so advantage is
negligible.
Applying triangle inequality over the two hops yields the theorem bound.
\end{proof}

\begin{theorem}[Mapping Integrity (Primary Instantiation: Option~B)]
\label{thm:map}
Let the registry enforce signature verification under
$\mathsf{pk}_{\mathrm{owner}}$ with an EUF-CMA secure scheme, and enforce
initialize-if-empty semantics that atomically bind
$\mathsf{pk}_{\mathrm{owner}}$ on first registration of each $\DID$
through the commit--reveal rule of Section~\ref{sec:authopt}, with
honest transactions included within $\Delta$ blocks.
For owner-key provisioning, consider:
\begin{itemize}[nosep]
  \item[(R)] \emph{Random-key model:} $\mathsf{sk}_{\mathrm{owner}}$ is generated uniformly and stored securely.
  \item[(D)] \emph{Passphrase-derived model:} $\mathsf{sk}_{\mathrm{owner}}$ is deterministically derived from passphrase-derived key material.
\end{itemize}
Then for any PPT adversary~$\mathcal{A}$ playing
$\mathbf{G}_{\mathrm{map}}$,
\[
  \mathrm{Adv}^{\mathrm{map}}_{\mathcal{A}}
  \;\leq\;
  \begin{cases}
    \mathrm{Adv}^{\mathrm{euf\text{-}cma}}_{\mathsf{Sig}}
    + \frac{q_A}{2^{\mu}} + \mathsf{negl}(\lambda)
    & \text{(R)}, \\
    \mathrm{Adv}^{\mathrm{euf\text{-}cma}}_{\mathsf{Sig}}
    + \mathrm{Adv}^{\mathrm{ow}}_{\mathsf{Argon2id}}
    + \frac{q_A}{2^{\mu}}
    + \mathsf{negl}(\lambda)
    & \text{(D)},
  \end{cases}
\]
where $\mu$ is the passphrase min-entropy and the $q_A/2^{\mu}$ term
bounds computing $\DID^*$ from $(I^*,P)$
before the honest reveal.
An unauthorized mapping change therefore requires a signature forgery or
a correct passphrase guess.
\end{theorem}

\begin{proof}[Proof sketch with reduction structure]
Suppose $\mathcal{A}$ causes an unauthorized accepted mutation on $\DID^*$.
Construct a reduction that embeds a challenge verification key as
$\mathsf{pk}_{\mathrm{owner}}^*$ and simulates all allowed interactions.
Because acceptance requires a valid signature on
$(\mathtt{op},\DID,\mathtt{cid},\mathtt{ver},\Commit)$ under
$\mathsf{pk}_{\mathrm{owner}}^*$, any successful unauthorized mutation
induces either:
(i) a fresh EUF-CMA forgery, or
(ii) in model~(D), successful recovery of passphrase-derived owner key
material via guessing.
For first registration, an attacker $\texttt{init}$ on $\DID^*$ with an
attacker key needs a commitment to $\DID^*$ that is at least $\Delta$
blocks old when the honest reveal is included.
Producing it before the honest reveal requires either computing $\DID^*$,
which in the random-oracle model for $\mathsf{Argon2id}$ requires querying
$(P,\mathsf{salt}_{\mathrm{lookup}})$ (probability at most
$q_A/2^{\mu}$) or guessing the pseudorandom $\DID^*$ (negligible), or
opening the honest commitment $c$ to a different key, which requires a
collision in $H$ (negligible).
A commitment made after the honest reveal is broadcast is younger than
$\Delta$ blocks when that reveal is included, so it cannot be accepted
first.
Therefore
$\mathrm{Adv}^{\mathrm{map}}_{\mathcal{A}}$ is bounded by the sum of EUF-CMA
advantage, the first-registration guessing term $q_A/2^{\mu}$, and (for
model~(D)) Argon2id one-way/passphrase-guessing advantage, up to
negligible terms.
\end{proof}

\begin{remark}[Extended Instantiation (Out of Core Scope): Option~A]
If a concrete verifier model supports secure validation of
HMAC-based authorization under $K_{\mathrm{reg}}$ without exposing secret
material (e.g., sound ZK validity proofs with binding statements), then a
parallel mapping-integrity bound can be derived from Argon2id one-wayness
and HMAC PRF security.
\end{remark}

\begin{theorem}[Rollback / Tamper Safety Under Fresh Finalized Reads]
\label{thm:roll}
If the registry verifier enforces strictly increasing version numbers
and the client verifies $H(\SA2) = \Commit$ upon recovery
\emph{while reading a fresh finalized registry state}, then for
any PPT adversary~$\mathcal{A}$ playing $\mathbf{G}_{\mathrm{roll}}$,
\[
  \exists\ \text{PPT reducer } \mathcal{B}\ \text{such that}
\]
\[
  \mathrm{Adv}^{\mathrm{roll}}_{\mathcal{A}}
  \;\leq\;
  \mathrm{Adv}^{\mathrm{cr}}_{H}
  \;+\;
  \mathrm{Adv}^{\mathrm{map}}_{\mathcal{B}}
  \;+\;
  \varepsilon_{\mathrm{fresh}}
  \;+\; \mathsf{negl}(\lambda),
\]
where $\varepsilon_{\mathrm{fresh}}$ is the probability that the client reads
an unfinalized or stale registry view despite the deployment's freshness
policy.
Tampering is detected except with negligible hash-collision probability,
and stale replay requires violating the fresh-finalized-read assumption or
forging authorization.
\end{theorem}

\begin{proof}[Proof sketch with event decomposition]
\emph{Version rollback.}
The registry rejects any update with
$\mathtt{ver}_{\mathrm{new}} \leq \mathtt{ver}_{\mathrm{old}}$.
Thus the adversary cannot revert the mapping to an older version via
a legitimate registry write.
Presenting a stale record to the client requires either controlling
the client's registry view so that freshness/finality assumptions fail,
or forging authorization for a downgrade, which
contradicts Theorem~\ref{thm:map}.
Define events:
$E_1$ (stale/non-finalized registry view at client),
$E_2$ (unauthorized downgrade write accepted),
$E_3$ (tampered artifact accepted).
Event $E_1$ contributes $\varepsilon_{\mathrm{fresh}}$ by assumption.
Event $E_2$ implies a break of mapping integrity and is bounded by
$\mathrm{Adv}^{\mathrm{map}}_{\mathcal{B}}$.

\emph{Artifact tampering.}
If $\mathcal{A}$ substitutes $\SA2' \neq \SA2$ while keeping the
same $\CID$ or $\Commit$, the client's commitment check
$H(\SA2') = \Commit$ fails unless $H(\SA2') = H(\SA2)$, which
requires a collision in~$H$, bounding $E_3$ by
$\mathrm{Adv}^{\mathrm{cr}}_H$.
A union bound over $E_1 \cup E_2 \cup E_3$ yields the claimed result.
\end{proof}

\begin{corollary}[Recoverability --- G1]
Since lookup keys are computable from $(I,P)$ alone
(Eq.~\eqref{eq:klookup}) and sealing keys are derivable after fetching
$\SA2$ from embedded $(\mathsf{salt}_{pw}, \mathsf{params}_{pw})$
(Table~\ref{tab:sa2}), any new device holding $(I, P)$ can recompute
$\DID$, fetch $\SA2$, verify $\Commit$, and decrypt, recovering $\REV$.
No device-bound secret is required, satisfying~G1.
\end{corollary}

\section{Privacy Considerations}
\label{sec:privacy}

\subsection{Public Observability}
The registry is public.
Even with keyed discovery identifiers, a motivated attacker might:
\begin{itemize}[nosep,leftmargin=1.5em]
  \item Observe updates over time if they already know a user's $\DID$.
  \item Correlate update times or CIDs across networks.
\end{itemize}
Implementations should minimize metadata and avoid embedding identifying
information in $\SA2$'s associated data field.

\subsection{Identifier Normalization Leakage}
Normalization must be deterministic and must not embed additional semantics.
Raw identifiers must never be stored on-chain.
Only the keyed discovery identifier is published.

\subsection{Migration Pointer Linkability}
If an on-chain redirect pointer is used during passphrase rotation
(Flow~3, Case~2), the link between the old $\DID$ and the new $\DID'$
becomes publicly observable, allowing an observer to correlate the two
registry entries as belonging to the same user.
Implementations should weigh this privacy cost against the UX benefit of
migration hints; where unlinkability is required, the old entry should be
tombstoned without a forward pointer.
Omitting the pointer is not sufficient on its own: the new entry must use
a fresh $\mathsf{pk}_{\mathrm{owner}}$, and the tombstone and the new
registration must not be linkable through a shared sender or relayer
account or through submission timing; otherwise the two entries remain
linkable.

\subsection{Offline Guessing}
If an adversary obtains $\SA2$ ciphertext, they may attempt offline
password guessing against~$P$.
This risk is inherent to password-based recovery.
$\mathsf{Argon2id}$ parameters should be set to high memory cost
appropriate for mobile performance budgets.
Optional secondary secrets (e.g., a Secondary Passphrase) can strengthen
recovery.

\section{Availability Considerations}
\label{sec:avail}

\subsection{Decentralized Storage}
Content-addressed storage does not guarantee persistence.
For recoverability, $\SA2$ must be available:
\begin{itemize}[nosep,leftmargin=1.5em]
  \item Users may pin their own $\SA2$.
  \item Systems may support multiple storage backends (multiple
        CIDs/replicas).
  \item Wallets may allow optional third-party pinning providers,
        acknowledging this reintroduces a replaceable centralized
        component.
\end{itemize}

\subsection{Registry Availability}
If the registry is on an L2, outages or censorship can block recovery.
Clients should support reading from multiple RPC endpoints and may support
mirrored registries.

\section{Revocation Semantics}
\label{sec:revocation}

Updating $\DID \to \mathtt{cid}$ (or rotating to a new entry) achieves:
\begin{itemize}[nosep,leftmargin=1.5em]
  \item \textbf{Discovery revocation:} new devices will retrieve the latest
        $\SA2$, not the old one.
\end{itemize}
It does \emph{not} automatically achieve:
\begin{itemize}[nosep,leftmargin=1.5em]
  \item \textbf{On-chain key revocation:} previously derived keys may still
        control assets on existing addresses.
\end{itemize}
True key revocation requires account migration, smart-account key rotation,
or asset transfer to new derived addresses.
Likewise, a superseded $\SA2$ remains permanently retrievable from
Arweave and still opens under its old passphrase. Rotation therefore does not
remediate an exposure that has already occurred: once $P$ is believed exposed, a
fresh $\REV$ must be generated and the assets moved to addresses derived from it,
since other addresses under the same $\REV$ remain exposed.

\section{Implementation Notes}
\label{sec:impl}

\paragraph{Parameter selection.}
The two-stage derivation (Section~\ref{sec:keys}) requires two
$\mathsf{Argon2id}$ evaluations per registration or recovery.
Because the public $\DID$ lets an adversary confirm a passphrase guess
with Stage~A alone, $\mathsf{params}_{\mathrm{lookup}}$ must be at least
as costly as $\mathsf{params}_{pw}$; a lighter Stage~A would lower the
cost of recovering $\REV$ to the Stage~A cost.
Concrete parameter sets may follow RFC~9106~\cite{rfc9106} or platform-specific guidance.
As a deployment baseline, Stage~A should not be configured below
approximately $64$~MiB memory and should target at least
$80$--$150$~ms on the target client class; Stage~B should be set at least as
strong as Stage~A (typically stronger).
Operationally, deployments should publish a target offline-guessing cost
floor (time and memory per guess) so security claims are auditable.
AEAD nonces must be unique per key; consider storing a random nonce
in~$\SA2$.
All sensitive values ($\REV$, derived keys) must be zeroized after use.

\paragraph{Passkey PRF usage.}
PRF output should never be exported.
Use it only to derive local sealing keys.
Local artifacts should be stored in the OS secure enclave/keychain as
appropriate.

\paragraph{DoS controls.}
Because there is no centralized server, DoS mitigations are limited.
Clients may cache successful registry lookups, use multiple
Arweave gateways, and support multiple redundant CIDs.

\paragraph{Canonical registry binding.}
No mechanism can prevent third parties from deploying parallel registries
or reusing public metadata such as $\mathtt{app\_id}$.
Deployments should therefore publish a signed configuration (or client
hardcoded allowlist) that binds the canonical
$(\mathtt{app\_id}, \mathtt{registry\_id})$, where
$\mathtt{registry\_id}$ includes at least
$(\mathtt{chain\_id}, \mathtt{contract\_address})$.
Clients should reject recovery/registration against registries outside this
authenticated binding.

\paragraph{PQC readiness and migration.}
Because \VADAR is layered over ACE-GF, it inherits algorithm agility at the
identity/key-derivation layer via explicit context separation (including
algorithm identifiers). Classical and post-quantum key
domains can coexist under the same identity root without re-enrollment.
For full end-to-end PQ deployment, the registry-authorization layer must also
be instantiated with PQ-secure signatures in Option~B; this paper's security
analysis remains generic in terms of signature EUF-CMA assumptions.

\section{Complexity and Cost}
\label{sec:cost}

\paragraph{Client cost.}
Dominated by two memory-hard derivations ($\mathsf{Argon2id}$, one per
stage) and a single AEAD operation.
All other operations (HKDF, HMAC, signature) are negligible in
comparison.

\paragraph{Offline-guessing cost model.}
Let $T_A$ be the measured time of Stage~A (lookup) and $T_B$ the measured
time of Stage~B (artifact key), on the attacker hardware model under
deployment parameters.
Let passphrase min-entropy lower bound be $\mu$.
An adversary who knows $I$ can test each passphrase guess against the
public $\DID$ with Stage~A alone, and needs Stage~B only once, for the
guess that matches.  Expected work for exhaustive guessing is therefore
approximately $2^{\mu-1}T_A + T_B$ (with memory footprint determined by
the Stage~A Argon2id settings), so Stage~A, not Stage~B, sets the cost of
recovering $\REV$ from public data.
This relates the policy ($\mu$ and Argon2id parameters) to the practical
brute-force cost in a form that can be audited.

\paragraph{Registry cost.}
Constant-size key lookup and a single signature verification (Option~B)
per mutation, plus one stored commitment and one hash check for
initialization.
For Option~A, verification cost is verifier-model dependent:
off-chain verifiers perform keyed checks directly, while on-chain
deployments typically require auxiliary mechanisms (e.g., commit-reveal
or ZK validity proofs), so no single ``one HMAC check'' on-chain cost
claim is assumed.
Storage is $O(1)$ per user.

\paragraph{Data availability.}
Reads are public and vendor-agnostic; write relayers are optional
convenience components.
A relayer cannot forge authorization, so it does not affect mapping
integrity for updates.
A relayer that withholds an honest reveal for more than $\Delta$ blocks
violates the inclusion assumption of Theorem~\ref{thm:map} and can squat
the first registration; the client then detects the rejected reveal at
registration time, so this is a registration-time denial of service, and
clients should submit the reveal directly or through more than one
relayer.

\section{Limitations}
\label{sec:limitations}

\begin{itemize}[nosep,leftmargin=1.5em]
  \item Weak passphrases degrade G2 and offline guessing resistance;
        policy must enforce entropy and KDF cost.
  \item Lost passphrase remains catastrophic for purely passphrase-based
        portability.
  \item Availability depends on registry and storage reachability;
        multi-endpoint fallback is required.
  \item If the owner's signing key is compromised (Option~B), rollback
        resistance fails; this is inherent.
  \item In model~(D), mapping integrity and enumeration resistance are
        both bounded by passphrase-guessing hardness; a single passphrase
        compromise defeats both goals simultaneously.
\end{itemize}

\section{Future Work}
\label{sec:future}

\begin{itemize}[nosep,leftmargin=1.5em]
  \item \textbf{On-chain authorization without revealing secrets:}
    ZK proofs for passphrase-derived authorization (Option~A).
  \item \textbf{Anti-correlation mechanisms:}
    Private information retrieval (PIR) or OPRF-based registry access.
  \item \textbf{True key revocation:}
    Integration with smart accounts for rotation and social recovery.
  \item \textbf{Multi-identifier support:}
    Mapping phone/email aliases with privacy-preserving linkage.
\end{itemize}

\section{Conclusion}
\label{sec:conclusion}

\VADAR supports decentralized recovery while making public-directory
discovery keyed and passphrase-gated.
The protocol separates three artifact roles, namely local (PRF-sealed), backup
(password-sealed), and recovered-local, so that device-bound security
enhances daily use without creating vendor lock-in for recovery.
Domain-separated key derivation, keyed discovery identifiers, and
commitment-bound versioned records are the mechanisms behind
enumeration resistance, mapping integrity, and rollback protection.
Under standard assumptions, \VADAR meets its security goals while
remaining vendor-agnostic and chain-agnostic.
End-to-end post-quantum deployment remains possible once registry
authorization (Option~B) is instantiated with PQ-secure signatures.

\section*{Acknowledgments}
This document builds on the ACE-GF artifact model and extends it with a
decentralized discovery and recovery layer.


\end{document}